\documentclass[12pt]{amsart}
\usepackage[margin=1in]{geometry}
\usepackage[english]{babel}
\usepackage[utf8]{inputenc}
\usepackage[all]{xy}
\usepackage[labelfont=bf,labelsep=period,justification=raggedright]{caption}
\usepackage[colorinlistoftodos]{todonotes}
\usetikzlibrary{calc}

\usepackage{fancyhdr}
\usepackage[font=small,labelfont=bf]{caption}

\usepackage[mathscr]{eucal}
\usepackage{amsaddr,multirow, amsmath,amssymb,amsthm,asymptote,wrapfig,dsfont,amsfonts,hyperref,fancyhdr,listings,csquotes,listings,xurl,subcaption,amsfonts,mathrsfs,hyperref,tkz-fct,tikz,multicol,bm,svg,textcomp,graphicx, color,indent first,hyperref,tkz-fct,tikz,xcolor,soul,bbold,stmaryrd,algorithm,algpseudocode,tkz-fct,tikz,tikz-cd,tabularx,booktabs,float,hyperref,url,caption,threeparttable,rotating,setspace,enumitem, array, cite}

\newcolumntype{Y}{>{\centering\arraybackslash}X}

\setlist[enumerate]{after=\vspace{-0.5\baselineskip}}

\usepackage[colorinlistoftodos]{todonotes}
\newcommand{\mycomment}[1]{}

\PassOptionsToPackage{dvipsnames,svgnames}{xcolor}

\theoremstyle{plain}
\newtheorem{thm}{Theorem}
\newtheorem{lemma}[thm]{Lemma}

\theoremstyle{definition}
\newtheorem{defn}[thm]{Definition}

\theoremstyle{remark}

\numberwithin{equation}{section}
\numberwithin{thm}{section}

\title{Improving Efficiency in Federated Learning with Optimized Homomorphic Encryption}
\author[Feiran Yang]{Feiran Yang}

\begin{document}

\pagestyle{plain}
\setcounter{page}{1}

\maketitle

\begin{abstract}

\indent Federated learning is a method used in machine learning to allow multiple devices to work together on a model without sharing their private data. Each participant keeps their private data on their system and trains a local model and only sends updates to a central server, which combines these updates to improve the overall model. A key enabler of privacy in FL is homomorphic encryption (HE). HE allows computations to be performed directly on encrypted data. While HE offers strong privacy guarantees, it is computationally intensive, leading to significant latency and scalability issues—particularly for large-scale models like BERT. In my research, I aimed to address this inefficiency problem. My research introduces a novel algorithm to address these inefficiencies while maintaining robust privacy guarantees. I integrated several mathematical techniques such as selective parameter encryption, sensitivity maps, and differential privacy noise within my algorithms, which has already improved its efficiency. I have also conducted rigorous mathematical proofs to validate the correctness and robustness of the approach. I implemented this algorithm by coding it in C++, simulating the environment of federated learning on large-scale models, and verified that the efficiency of my algorithm is $3$ times the efficiency of the state-of-the-art method. This research has significant implications for machine learning because its ability to improve efficiency while balancing privacy makes it a practical solution! It would enable federated learning to be used very efficiently and deployed in various resource-constrained environments, as this research provides a novel solution to one of the key challenges in federated learning: the inefficiency of homomorphic encryption, as my new algorithm is able to enhance the scalability and resource efficiency of FL while maintaining robust privacy guarantees. \\

\end{abstract}

\newpage
\tableofcontents
\newpage

\section{Introduction}\label{sec:introduction}

Federated learning (FL) became increasingly popular in distributed systems and privacy-preserving machine learning, due to its ability to enable collaborative model training across decentralized datasets without directly sharing sensitive information. In FL, clients train models locally and share encrypted updates with a central server for aggregation, thus preserving data privacy. Unlike traditional centralized machine learning, which requires collecting data in one location, FL keeps data decentralized and ensures that they stay in original locations. 

However, privacy vulnerabilities persist, as malicious servers may exploit aggregated updates to reconstruct sensitive data or infer private information, since the shared model is updated accordingly with the private data. To address these threats, homomorphic encryption (HE) has been established, allowing computations on encrypted data without decryption. Despite its capability to ensure data privacy, HE faces significant challenges due to its high computational and communication overheads, which limit its scalability and feasibility in federated learning systems. In particular, HE-based FL faces challenges like high latency, increased energy consumption, and resource demands, especially in large-scale models or with limited computational capacity. Traditional HE solutions fail to adequately optimize for the decentralization of FL, making the process inefficient for large foundation models like ResNet and BERT. Existing methods such as selective parameter encryption and adaptive HE strategies still struggle with balancing privacy preservation and computational efficiency. 

This paper introduces a new algorithm to improve the efficiency of homomorphic encryption in federated learning systems. 

\section{Related Works}

Federated Learning (FL) has revolutionized the way organizations collaborate on model training without exposing private data. A seminal contribution was the Federated Averaging (FedAvg) algorithm by McMahan et al. \cite{mcmahan2017communication}, which demonstrated how deep networks could be trained efficiently on decentralized data across mobile devices. As the technology matured, researchers became increasingly concerned with privacy and security vulnerabilities, sparking the development of a rich body of defense strategies. 

One prominent privacy concern is the leakage of training data through shared gradients. Wei et al. \cite{wei2020gradient} underscored this vulnerability by illustrating how adversaries could exploit gradients to reconstruct sensitive information. To combat such attacks, defenses like Soteria \cite{sun2021soteria} introduced random perturbations to data representations. Meanwhile, Fed-CDP \cite{wei2021gradient} harnessed client-level differential privacy to further protect gradients without significantly compromising model performance.

Numerous privacy-preserving approaches have combined cryptographic and statistical methods to secure FL. Truex et al. \cite{truex2019hybrid} introduced a hybrid framework that blends Secure Multiparty Computation (SMC) and Differential Privacy (DP) for balanced security and scalability. Xu and Ma \cite{xu2019hybridalpha} explored functional encryption in their HybridAlpha solution, aiming to secure FL workflows end to end. Concurrently, Liu et al. \cite{liu2022privacy} emphasized the importance of Privacy-Preserving Aggregation (PPAgg) protocols to safeguard model updates.

Several comprehensive surveys categorize and analyze these techniques. Yin et al. \cite{yin2021comprehensive} offer a broad taxonomy of privacy-preserving methods, while Lyu et al. \cite{lyu2020threats} delve deeper into attacks such as model poisoning and inference attacks. Jiang et al. \cite{jiang2022privacy} extended these discussions into the realm of Vertical Federated Learning (VFL), identifying unique risks at the prediction stage. Additionally, Zhang et al. \cite{zhang2022security} explored integrating blockchain and Trusted Execution Environments (TEEs) to bolster FL security.

In industrial and healthcare environments, FL faces both resource and data sensitivity constraints. Luo et al. \cite{luo2019efficient} developed frameworks tailored to industrial AI systems, addressing communication bottlenecks and privacy requirements. Wei et al. \cite{wei2020differential} showed how differential privacy mechanisms could be optimized in constrained environments. More recently, Hu et al. \cite{hu2024overview} surveyed the latest FL security developments, including methods to deal with data heterogeneity, adversarial robustness, and communication overhead.

Emerging defense strategies focus on mitigating gradient leakage while maintaining model performance. NbAFL \cite{wei2020mechanisms} introduced adaptive noise injection, striking a balance between accuracy and privacy. Similarly, Sun et al. \cite{sun2021defense} proposed more targeted perturbations to data representations, reinforcing resistance to gradient-based attacks. Zhang et al. \cite{zhang2022threats} advanced the field further by investigating homomorphic encryption (HE) and blockchain-based solutions for secure aggregation.

Despite considerable progress, homomorphic encryption remains a challenging bottleneck in large-scale FL deployments. Fully homomorphic schemes often require expensive polynomial operations, frequent relinearization or bootstrapping, and considerable memory for storing encrypted parameters. Although selective encryption methods can reduce overhead, they risk compromising privacy by leaving portions of the data unprotected \cite{zhang2022threats}. Conversely, fully encrypted solutions sometimes become infeasible in real-time or large-scale scenarios due to the sheer computational load.

Addressing these constraints calls for more efficient HE protocols and hybrid solutions that combine cryptography, differential privacy, and robust architectural designs. As FL increasingly supports critical domains such as healthcare \cite{jiang2022privacy} and industrial automation \cite{luo2019efficient}, research must continue refining homomorphic encryption techniques to ensure both strong privacy guarantees and practical runtime performance.

\section{Preliminaries}\label{2}

\subsection{Federated Learning} \label{fl}

\begin{defn}(\textbf{Federated Learning})
Federated learning is a privacy-preserving framework where multiple clients each hold their own private dataset but collectively wish to train a shared model. The FL process proceeds in rounds:
\begin{enumerate}
    \item \textit{Global Model Initialization:} The server initializes a global model \(M^{(0)}\).
    \item \textit{Local Training:} Each participating client \(C_i\) downloads the global model \(M^{(t)}\) (from the previous round \(t\)) and trains it locally with its private data for a fixed number of epochs or until a convergence criterion is met. 
    \item \textit{Upload Encrypted Updates:} Each client encrypts its local model update \(\Delta_i^{(t)}\) using a homomorphic encryption scheme (Definition~\ref{he_def}) and sends the encrypted update to the server.
    \item \textit{Aggregation:} The server, without decrypting the data, homomorphically aggregates the local updates \(\Delta_i^{(t)}\). 
    \item \textit{Model Update:} The server updates the global model parameters \(M^{(t+1)} = M^{(t)} + \text{HE.Aggregate}\{\Delta_i^{(t)}\}\).
    \item \textit{Repeat or Terminate:} The procedure continues for another round until the global model converges or a fixed number of iterations is completed.
\end{enumerate}
\end{defn}

\begin{defn}(\textbf{Local Objective Function in FL})
In federated learning, each client \(C_i\) holds a local dataset \(\mathcal{D}_i\) and aims to minimize an objective function
\[
F_i(\mathbf{w}) \;=\; \frac{1}{|\mathcal{D}_i|} \sum_{(x,y)\in\mathcal{D}_i} \ell(\mathbf{w}; x, y),
\]
where \(\ell\) is a loss function (e.g., cross-entropy for classification). The global objective is often expressed as a weighted sum of local objectives:
\[
F(\mathbf{w}) \;=\; \sum_{i=1}^{N} \pi_i \, F_i(\mathbf{w}),
\]
where \(\pi_i = \frac{|\mathcal{D}_i|}{\sum_{j=1}^{N} |\mathcal{D}_j|}\) or a similar weighting scheme.
\end{defn}

\begin{defn}(\textbf{Non-IID Data})
A common challenge in FL is that clients may have non-identically and independently distributed (\textit{non-IID}) data. Formally, each \(\mathcal{D}_i\) is drawn from a (potentially) different underlying distribution \( \mathcal{P}_i \). The aggregation step must account for these heterogeneous distributions to avoid bias in the global model.
\end{defn}

\begin{defn}(\textbf{Threat Model in FL})
We consider a semi-honest (also called \emph{honest-but-curious}) server or adversary who follows the protocol correctly but attempts to infer additional information about the clients’ data from the intercepted messages. Adversaries may also compromise a subset of clients, gaining access to their local updates or keys. This motivates the use of cryptographic techniques (e.g., homomorphic encryption) and privacy mechanisms (e.g., differential privacy).
\end{defn}

Federated learning mitigates data privacy risks by preventing raw data from leaving local devices. Nevertheless, partial leakage may still occur through shared gradient updates, necessitating additional cryptographic techniques to ensure privacy.

\subsection{Homomorphic Encryption Scheme} \label{he}

\begin{defn}(\textbf{Partially vs. Fully Homomorphic Encryption})
A homomorphic encryption scheme is called:
\begin{itemize}
    \item \textit{Partially Homomorphic (PHE)} if it supports homomorphic evaluation of either addition \emph{or} multiplication (but not both arbitrarily).
    \item \textit{Somewhat/Fully Homomorphic (SHE/FHE)} if it supports an unbounded number of both additions and multiplications on ciphertexts (fully) or supports them up to a certain circuit depth (somewhat).
\end{itemize}
In federated learning, many practical protocols rely on partially or somewhat homomorphic schemes for efficient encrypted aggregation (e.g., additive homomorphisms to sum encrypted gradients).
\end{defn}

\begin{defn}(\textbf{Homomorphic Encryption})
It is a cryptographic technique that enables computation on ciphertexts as if it were plain data. If \(Enc(\cdot)\) is our encryption function and \(\oplus\) is a homomorphic operation (such as addition), we want the property that:
\[
Enc(a) \oplus Enc(b) = Enc(a + b)
\]
for (fully or partially) homomorphic schemes. In FL, this property allows the central server to sum or average the encrypted model updates from the clients without decrypting. 
\end{defn}

\begin{defn} \label{he_def}
(\textbf{Homomorphic Encryption Scheme}). Let $\lambda$ be a security parameter. A homomorphic encryption scheme $\Pi$ consists of:
\begin{itemize}
    \item $\textbf{KeyGen}(\lambda) \rightarrow (pk, sk, ek)$: Generates a public key $pk$, secret key $sk$, and evaluation key $ek$.
    \item $\textbf{Enc}(pk,m) \rightarrow ct$: Encrypts a plaintext message $m$ using the public key $pk$ and outputs ciphertext $ct$.
    \item $\textbf{Dec}(sk,ct) \rightarrow m$: Decrypts a ciphertext $ct$ using the secret key $sk$ to recover the plaintext $m$.
    \item $\textbf{Eval}(ek,\circ,ct_1,\dots,ct_n) \rightarrow ct_{\text{eval}}$: Given an evaluation key $ek$, a circuit (or arithmetic operation) $\circ$, and ciphertexts $ct_1,\dots,ct_n$, outputs a ciphertext $ct_{\text{eval}}$ such that 
    \[
      \textbf{Dec}(sk, ct_{\text{eval}}) = \circ(\textbf{Dec}(sk,ct_1), \dots, \textbf{Dec}(sk,ct_n)).
    \]
\end{itemize}
\end{defn}

\begin{defn}(\textbf{Noise Budget in HE})
Most homomorphic encryption schemes rely on a noise term introduced during \textbf{Enc} to ensure security. Each homomorphic operation can grow this noise. When the noise exceeds a certain threshold, decryption fails or produces an incorrect result. The \emph{noise budget} refers to the capacity of a ciphertext to tolerate further homomorphic operations before exceeding this threshold.
\end{defn}

HE is well-suited for federated learning scenarios where clients encrypt their local updates before sending them to the server. However, the computational overhead grows significantly for high-dimensional models and large-scale neural networks.

\subsubsection{Sensitivity Map}

In many learning tasks, model parameters contribute differently to the overall performance or carry different levels of sensitive information. A \emph{sensitivity map} helps quantify this variation.

\begin{defn} \label{sens_map_def}
(\textbf{Sensitivity Map}). Let $\mathbf{w} \in \mathbb{R}^d$ be the parameter vector of a machine learning model. A sensitivity map is a function 
\[
   S: \mathbb{R}^d \rightarrow \mathbb{R}^d
\]
that assigns to each parameter $w_j$ a value $S(\mathbf{w})_j \in \mathbb{R}$, indicating how sensitive or privacy-critical $w_j$ is. A larger value of $S(\mathbf{w})_j$ suggests a higher sensitivity level for the parameter $w_j$.
\end{defn}

\begin{lemma} \label{monotonic_lemma}
(\textbf{Monotonic Mapping Property}). Suppose the sensitivity map $S(\mathbf{w})_j$ is monotonically related to a risk measure $\rho(\mathbf{w})_j$ that captures privacy or vulnerability (e.g., gradient magnitude, personal information density). Then for any scalar $c \geq 1$, we have:
\[
   S(\mathbf{w})_j \leq c \, \rho(\mathbf{w})_j \quad \forall j.
\]
\end{lemma}
\begin{proof}
The proof follows from the definition of monotonicity: $S(\mathbf{w})_j$ is bounded by a constant factor times the risk measure if $S(\mathbf{w})_j$ is derived via a monotonic transformation of $\rho(\mathbf{w})_j$. 
\end{proof}

\begin{defn}(\textbf{Sensitivity Thresholding})
Given a sensitivity map \(S(\mathbf{w})\) and a user-defined threshold \(\tau\), define:
\[
   \mathcal{I}_{\text{enc}}(\tau) \;=\; \{\, j \,\mid\, S(\mathbf{w})_j > \tau \}, 
   \quad
   \mathcal{I}_{\text{plain}}(\tau) \;=\; \{\, j \,\mid\, S(\mathbf{w})_j \leq \tau \}.
\]
This partitioning plays a central role in \emph{selective encryption} of parameters.
\end{defn}

\subsubsection{Selective Parameter Encryption}

\begin{defn} \label{selective_enc_def}
(\textbf{Selective Parameter Encryption}). Let $\mathbf{w} \in \mathbb{R}^d$ and let $S(\mathbf{w})$ be a sensitivity map as in Definition~\ref{sens_map_def}. A \emph{selective parameter encryption} strategy $E$ partitions $\{1,\dots,d\}$ into two subsets: 
\[
  \mathcal{I}_{\text{enc}} = \{\,j \mid S(\mathbf{w})_j > \tau\}\quad\text{and}\quad
  \mathcal{I}_{\text{plain}} = \{\,j \mid S(\mathbf{w})_j \leq \tau\},
\]
for some threshold $\tau > 0$. Parameters indexed by $\mathcal{I}_{\text{enc}}$ are encrypted (e.g., via a homomorphic encryption scheme), while parameters indexed by $\mathcal{I}_{\text{plain}}$ are transmitted in plaintext or with a lighter security mechanism.
\end{defn}

\begin{lemma} \label{selective_enc_lemma}
(\textbf{Communication Reduction}). Assume that encrypting a parameter $w_j$ has cost $C_{\text{enc}} > 0$, whereas sending $w_j$ in plaintext has cost $C_{\text{plain}} \ll C_{\text{enc}}$. Under a selective parameter encryption strategy, the expected communication cost reduces to
\[
   |\mathcal{I}_{\text{enc}}|\cdot C_{\text{enc}} + |\mathcal{I}_{\text{plain}}|\cdot C_{\text{plain}},
\]
which is typically strictly less than encrypting all parameters (i.e., $d \cdot C_{\text{enc}}$) if $|\mathcal{I}_{\text{enc}}| < d$.
\end{lemma}
\begin{proof}
By partitioning $\{1,\dots,d\}$ based on $S(\mathbf{w})_j$ (Definition~\ref{selective_enc_def}), only a subset of parameters are fully encrypted. Summing costs over the partition yields the total communication cost. Comparisons with $d \cdot C_{\text{enc}}$ demonstrate reduction if \( |\mathcal{I}_{\text{enc}}| < d \).
\end{proof}

\subsubsection{Security Theorems}

\begin{thm}[Correctness of Homomorphic Encryption]
\label{thm:he_correctness}
Let \\ $\Pi = (\textbf{KeyGen}, \textbf{Enc}, \textbf{Dec}, \textbf{Eval})$ be a homomorphic encryption scheme with security parameter $\lambda$. Suppose $\circ$ is any arithmetic circuit (or function) over the message space. For all messages $m_1, m_2, \ldots, m_n$ in the valid plaintext space, for all keys $(pk, sk, ek) \leftarrow \textbf{KeyGen}(\lambda)$, and for ciphertexts $ct_i \leftarrow \textbf{Enc}(pk, m_i)$, the following holds with probability $1$ (or negligible error):
\[
   \textbf{Dec}\Bigl(sk,\; \textbf{Eval}\bigl(ek, \circ,\; ct_1,\ldots,ct_n\bigr)\Bigr)
   \;=\;\circ\bigl(m_1,\ldots,m_n\bigr).
\]
In other words, evaluating a circuit $\circ$ on the ciphertexts $ct_i$ and then decrypting yields the same result as applying $\circ$ on the plaintexts $m_i$ directly.
\end{thm}

\begin{thm}[Soundness of Homomorphic Encryption]
\label{thm:he_soundness}
Let $\Pi$ be as in Theorem~\ref{thm:he_correctness}, and assume $\Pi$ is \emph{IND-CPA} secure. Then for any probabilistic polynomial-time (PPT) adversary $\mathcal{A}$ that modifies or forges a ciphertext $ct^*$ in an attempt to change the decrypted plaintext in a nontrivial manner, we have that
\[
   \Pr\bigl[\textbf{Dec}(sk, ct^*) = m^* \land m^* \neq \text{ ``legitimate outcome'' } \bigr] 
   \;\le\; \text{negl}(\lambda).
\]
In other words, except with negligible probability, the adversary cannot produce or alter a ciphertext that decrypts to an unintended message. Soundness thus ensures that if $ct^*$ decrypts successfully, it corresponds to a valid homomorphic operation on previously encrypted messages (or decrypts to an invalid $\bot$).
\end{thm}

\begin{defn} \label{dp_def}
(\textbf{Differential Privacy}). A randomized algorithm $\mathcal{A}$ satisfies $(\epsilon, \delta)$-differential privacy if, for any two adjacent datasets $D$ and $D'$ (differing by at most one record), and for any set of possible outcomes $\mathcal{O}$,
\[
   \Pr[\mathcal{A}(D) \in \mathcal{O}] \;\le\; e^\epsilon \,\Pr[\mathcal{A}(D') \in \mathcal{O}] + \delta.
\]
\end{defn}

\begin{defn}(\textbf{Local vs. Global Differential Privacy})
\begin{itemize}
    \item \textit{Local DP}: Each client perturbs or adds noise to their data \emph{before} sending it to the server. The server sees only the noisy output, offering stronger privacy at the individual level but potentially lower accuracy.
    \item \textit{Global DP}: The noise is added centrally (e.g., by the server) to aggregate statistics or updates after collecting raw (or partially encrypted) data. This typically yields better utility but requires trust in the aggregator’s correct implementation.
\end{itemize}
\end{defn}

\begin{thm} \label{dp_composition}
(\textbf{Composition Theorem for Differential Privacy} \cite{dworkdp}). Suppose $\mathcal{A}_1, \mathcal{A}_2, \dots, \mathcal{A}_k$ are $k$ mechanisms, each satisfying $(\epsilon, \delta)$-DP. Then the composition $\mathcal{A} = (\mathcal{A}_1,\dots,\mathcal{A}_k)$ satisfies $\left(k\epsilon, k\delta\right)$-DP.
\end{thm}

\begin{defn}(\textbf{Gradient Clipping and DP Noise Injection})
A common DP mechanism in FL is to:
\begin{enumerate}
    \item \textit{Clip Gradients:} For each client gradient \(\nabla F_i(\mathbf{w})\), enforce \(\|\nabla F_i(\mathbf{w})\|\leq C\) by rescaling if necessary.
    \item \textit{Add Noise:} Perturb the clipped gradient with Gaussian or Laplacian noise: 
    \[
      \widetilde{\nabla F_i}(\mathbf{w}) = \nabla F_i(\mathbf{w}) + \mathcal{N}(0,\sigma^2 I).
    \]
\end{enumerate}
The clipping bounds and noise scales are chosen to satisfy $(\epsilon,\delta)$-DP under the composition theorem (Theorem~\ref{dp_composition}).
\end{defn}

\vspace{3mm}

\section{Framework} \label{3}

Here is the framework overview of the homomorphic encryption scheme and federated learning process. The algorithms and phases are in blue. 

\begin{figure}[H]
    \begin{center}
    \includegraphics[width=0.95\textwidth]{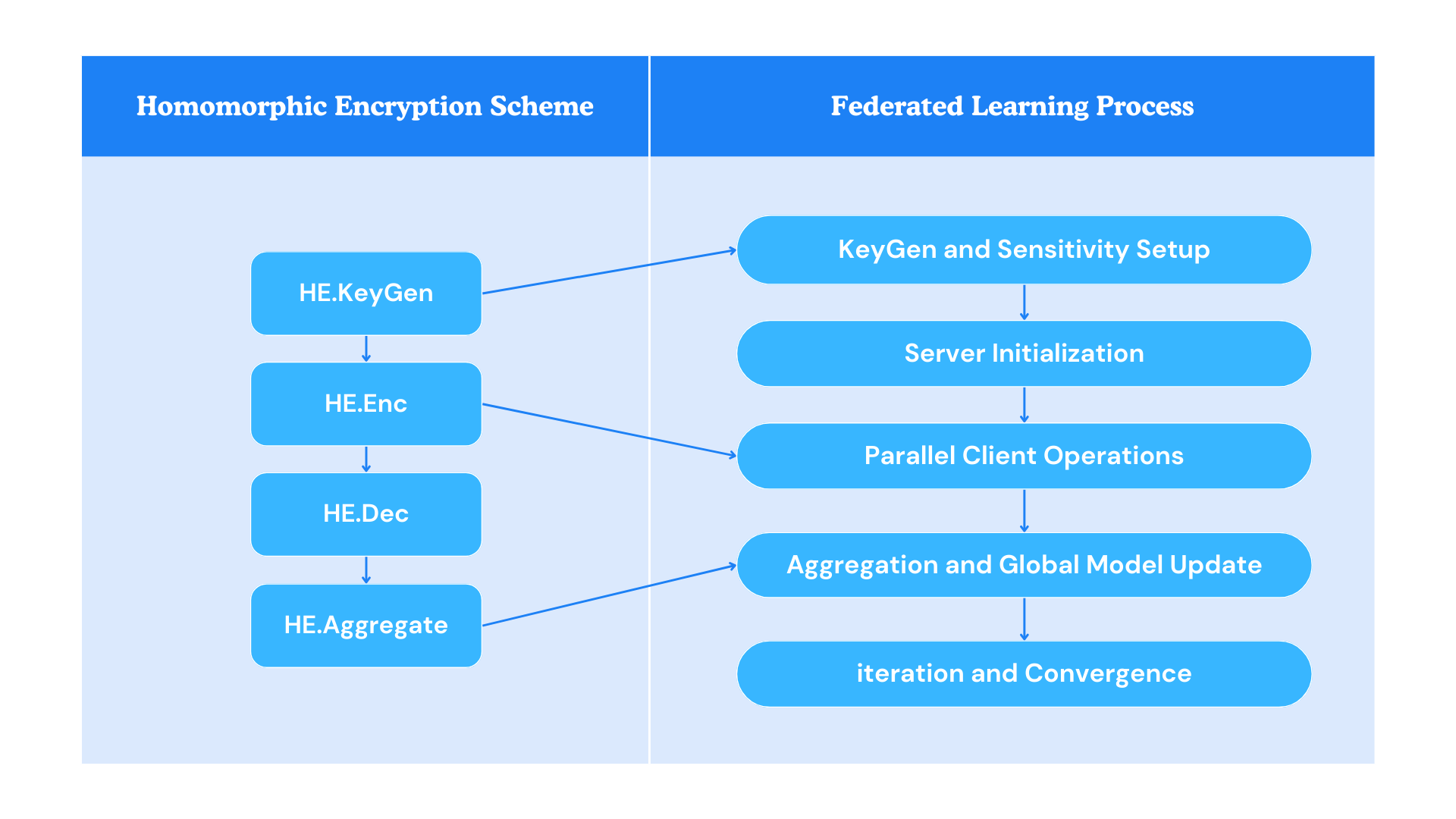}
    \caption{Framework of HES and Federated Learning}
    \end{center}
\end{figure}

\section{Algorithm}

\subsection{Main Idea}
We propose a specialized homomorphic encryption (HE) scheme that decreases the overhead of encryption and aggregation in FL while preserving strong privacy guarantees. Our approach integrates several mathematical and algorithmic optimizations:
\begin{itemize}
    \item \textbf{Selective Parameter Encryption:} Not all parameters of a neural network require the same level of precision or protection. We encrypt only sensitive or high-impact parameters at high precision, allowing us to skip heavy computations on parameters with low sensitivity.
    \item \textbf{Sensitivity Maps:} We create a sensitivity map that identifies which model parameters significantly impact performance. These parameters receive higher encryption security (and possibly differential privacy noise).
    \item \textbf{Embedded Differential Privacy:} We incorporate DP noise directly into the encrypted parameters based on the sensitivity map. This step maintains privacy even if partial decryption occurs, and it also limits the potential for reconstructing private information through repeated queries.
    \item \textbf{Optimized Ciphertext Packing and Batch Operations:} By leveraging packing techniques, we can bundle multiple model parameters into a single ciphertext. The result is that homomorphic additions or multiplications are performed in a “batch,” greatly reducing the total number of HE operations.
\end{itemize}

Empirically, our scheme achieves a $3\times$ speedup compared to the state-of-the-art while maintaining a high level of privacy protection, making FL viable in real-world, large-scale, and latency-sensitive applications.

\subsection{Construction} \label{heconstruction}
We define our homomorphic encryption (HE) scheme for federated learning (FL) as a tuple of algorithms:
\[
\bigl(\textbf{HE.KeyGen},\;\textbf{HE.Enc},\;\textbf{HE.Dec},\;\textbf{HE.Aggregate}\bigr),
\]
augmented by our specialized approach to \emph{partial encryption}, \emph{sensitivity mapping}, and \emph{embedded noise}.

\begin{itemize}

\item \textbf{HE.KeyGen($\lambda$)}:
\begin{itemize}
    \item Given a security parameter \(\lambda\), outputs a secret key \(sk\) and a public key \(pk\). The procedure is as follows: 
    \item Choose system parameters \((n, q, \chi)\) according to \(\lambda\), where \(n\) is a polynomial in \(\lambda\), \(q\) is a large modulus, and \(\chi\) is an error distribution.
    \item Sample a secret polynomial \(s(x)\) from \(\chi\) in the ring \(\mathcal{R} = \mathbb{Z}_q[x]/(f(x))\).
    \item Pick a random polynomial \(a(x)\) in \(\mathcal{R}\) and sample an error polynomial \(e(x)\) from \(\chi\).
    \item Set the public key as 
    \[
       pk = \bigl(a(x),\, b(x) = -(a(x)\,s(x)) \;-\; e(x)\bigr),
    \]
    and the secret key is 
    \[
       sk = s(x).
    \]
\end{itemize}

\item \textbf{HE.Enc($pk, \mathbf{m}$)} \(\to \mathbf{c}\):
\begin{itemize}
    \item Takes as input a public key \(pk\) and a vector of model parameters (or local updates) \(\mathbf{m} = (m_1, m_2, \dots, m_\ell)\).
    \item Outputs the ciphertext vector \(\mathbf{c}\).
    \item \emph{Partial Encryption}: Let \(\mathcal{I} \subseteq \{1, \ldots, \ell\}\) be the set of indices deemed “sensitive.” Only those coordinates in \(\mathcal{I}\) are encrypted:
    \[
      \widehat{m}_i =
      \begin{cases}
        \mathrm{Enc}_{pk}(m_i), & \text{if } i \in \mathcal{I}, \\
        m_i, & \text{otherwise}.
      \end{cases}
    \]
    \item \emph{Sensitivity Mapping}: A function 
    \(\mathrm{Sens}(\mathbf{m}) = (\sigma(m_1), \dots, \sigma(m_\ell))\) 
    can further guide which parameters get encrypted and whether additional noise is embedded.
\end{itemize}

\item \textbf{HE.Dec($sk, \mathbf{c}$)} \(\to \mathbf{m}\):
\begin{itemize}
    \item Takes as input the secret key \(sk\) and a ciphertext \(\mathbf{c}\).
    \item Outputs the decrypted model vector \(\mathbf{m}\).
    \item In our scheme, decryption is as follows:
    \[
      \widehat{m}(x) = c_2(x) + s(x)\, c_1(x) \mod q,
    \]
    which recovers the polynomial \(\widehat{m}(x)\). After scaling/unpacking, we obtain \(m\).
\end{itemize}

\item \textbf{HE.Aggregate($pk, \{\mathbf{c}_i\}$)} \(\to \mathbf{c}_{\text{agg}}\):
\begin{itemize}
    \item Takes a set of ciphertexts \(\{\mathbf{c}_1, \mathbf{c}_2, \dots, \mathbf{c}_n\}\) from \(n\) clients.
    \item Produces a single ciphertext \(\mathbf{c}_{\text{agg}}\) that represents the (homomorphic) aggregated update.
    \item \emph{Homomorphic Addition}: Denote by \(\oplus\) the homomorphic addition (component-wise for RLWE ciphertexts).
    \item \emph{Multiplying a Hash Function}: Let \(\mathcal{H}\colon \{\mathbf{c}\}\to \mathcal{R}_q\) be a cryptographic hash mapping each ciphertext \(\mathbf{c}_i\) to an element/polynomial in \(\mathcal{R}_q\). Denote homomorphic multiplication by \(\otimes\).
    \item Then:
    \[
      \mathbf{c}_{\text{agg}}
        \;=\;
        \bigoplus_{i=1}^n 
           \Bigl( \mathbf{c}_i \;\otimes\; \mathcal{H}(\mathbf{c}_i) \Bigr).
    \]
    \item Concretely:
    \begin{itemize}
        \item Compute the hash \(h_i = \mathcal{H}(\mathbf{c}_i)\).
        \item Homomorphically multiply each \(\mathbf{c}_i\) by \(h_i\):
        \[
          \mathbf{c}_i' = \mathbf{c}_i \otimes h_i.
        \]
        \item Sum over \(i\):
        \[
          \mathbf{c}_{\text{agg}}
            = \bigoplus_{i=1}^n \mathbf{c}_i'.
        \]
    \end{itemize}
\end{itemize}

\end{itemize}

\vspace{5mm}

Next, we describe the step-by-step workflow of our federated learning scheme that integrates homomorphic encryption (HE) and selective masking to maintain data privacy. The entities involved are a \textit{key authority}, a \textit{central server}, and multiple \textit{clients} that operate in parallel.

\vspace{3mm}

\textbf{1. Key Generation and Sensitivity Setup}

Generating Encryption Keys:
\begin{itemize}
    \item Generate the cryptographic tools required for homomorphic encryption.
    \item Run \(\texttt{HE.KeyGen}(\lambda)\) to produce \((\mathrm{PK}, \mathrm{SK})\). Optionally produce evaluation keys (\(\mathrm{EVK}\)) for partially/fully homomorphic operations, depending on the chosen HE scheme (e.g., BGV, BFV, CKKS).
\end{itemize}

Defining Sensitivity Levels and Collecting Metadata:
\begin{itemize}
    \item Assign each parameter or parameter group a ‘sensitivity level’ that dictates encryption precision and whether differential privacy (DP) noise is added.
    \item Collect metadata from clients (e.g., approximate data distributions, model architecture).
    \item Produce a vector \(\mathbf{v}_i\) for each client \(i\) indicating how sensitive each parameter group is. This can be manually assigned or learned via heuristics/analysis of gradient magnitudes.
    \item Partition model parameters using a sensitivity map \(\Gamma\):\\
    If \(\Gamma\) is our sensitivity map, we partition the model parameters \(\mathbf{m}\) into sub-vectors:
    \[
    \mathbf{m} = \mathbf{m}^H \cup \mathbf{m}^L,
    \]
    where \(\mathbf{m}^H\) denotes \textbf{highly sensitive parameters}, and \(\mathbf{m}^L\) denotes \textbf{low-sensitivity parameters}. We encrypt \(\mathbf{m}^H\) at high security levels (e.g., larger ciphertext modulus, deeper levels of homomorphic capacity) and possibly apply DP noise. Parameters in \(\mathbf{m}^L\) may be:
    \begin{itemize}
        \item Encrypted at a lower security level,
        \item Partially randomized, or
        \item Aggregated in the clear if it is proven that their compromise yields negligible information about the data.
    \end{itemize}
    
    \item Encrypt \(\mathbf{v}_i\) into \(\mathbf{V}_i = \texttt{HE.Enc}(\mathrm{PK}, \mathbf{v}_i)\) before sending to the server, ensuring that the server only has encrypted sensitivity data.
\end{itemize}

\vspace{3mm}
\textbf{2. Server Initialization with Sensitivity Maps}

Aggregating Sensitivity Vectors:
\begin{itemize}
    \item Combine client-specific sensitivity vectors into a global mask \(\mathbf{M}\).
    \item Homomorphically sum (or weighted sum) the encrypted vectors:
    \[
      \mathbf{S} \;=\; \sum_{i=1}^N \alpha_i \,\mathbf{V}_i.
    \]
    \item Apply a threshold \(\tau\) and a filter function \(\mathcal{F}\) to create a global mask \(\mathbf{M}\). For instance:
    \[
      \mathbf{M} = \mathcal{F}(\mathbf{S}, \tau).
    \]
    \(\mathcal{F}\) may set entries to ``highly protected'' if above \(\tau\), or to a lower/zero level otherwise.
\end{itemize}

Broadcasting the Encrypted Mask:
\begin{itemize}
    \item Provide each client with an encrypted representation \(\mathbf{M}\) that reveals no direct information about other clients’ sensitivities.
    \item Send \(\mathbf{M}\) to clients in encrypted form.
    \item Not decrypt \(\mathbf{M}\) on the server side; only clients (with \(\mathrm{SK}\)) can decrypt it.
\end{itemize}

\vspace{3mm}
\textbf{3. Parallel Client Operations with Selective Encryption and DP}

Decrypting the Mask Locally:
\begin{itemize}
    \item Let each client learn which parameters are ``high'' vs. ``medium/low'' sensitivity via local decryption.
    \item Perform \(\mathbf{M}_{\text{dec}} = \texttt{HE.Dec}(\mathrm{SK}, \mathbf{M})\).
    \item Locally interpret \(\mathbf{M}_{\text{dec}}\) to see how it overlaps with the client’s parameter structure.
\end{itemize}

Local Model Training:
\begin{itemize}
    \item Initialize or load the global model: \(\mathbf{W}_i^{(t)}\).
    \item Perform standard SGD or another optimizer on local dataset \(\mathcal{D}_i\).
    \item Obtain updated parameters \(\mathbf{W}_i^{(t+1)}\).
\end{itemize}

Injecting Differential Privacy Noise (Optional):
\begin{itemize}
    \item Obfuscate individual data contributions by adding noise correlating with parameter sensitivity.
    \item Determine noise variance \(\sigma\) or \(\Delta\) based on an \(\epsilon\)-DP budget and the sensitivity level.
    \item Add noise \(\boldsymbol{\eta}\) (e.g., Laplace/Gaussian) to \(\mathbf{W}_i^{(t+1)}\). More sensitive parameters receive larger noise.
\end{itemize}

Encrypting Sensitive Parameters Selectively:
\begin{itemize}
    \item Encrypt only the sensitive parts of \(\mathbf{W}_i^{(t+1)}\) at full precision; optionally compress or leave other parts in the clear.
    \item Split parameters:
    \[
      \mathbf{W}_i^{(t+1)} = \bigl(\mathbf{W}_{\text{sens}}, \mathbf{W}_{\text{nonsens}}\bigr).
    \]
    \item If the HE scheme supports multiple encryption levels:
    \begin{itemize}
        \item Use high precision ciphertext for HS parameters.
        \item Possibly lower precision ciphertext for MS parameters.
    \end{itemize}
    \item Leave LS parameters unencrypted, if policy allows.
    \item Form the update:
    \[
      \mathbf{U}_i = \bigl(\texttt{HE.Enc}(\mathrm{PK}, \mathbf{W}_{\text{sens}}), \; \mathbf{W}_{\text{nonsens}}\bigr).
    \]
    \item Use batching/packing to reduce ciphertext overhead if the HE scheme allows.
\end{itemize}

Uploading Updates to the Server:
\begin{itemize}
    \item Transmit partial or fully encrypted updates back to the server.
    \item Send \(\mathbf{U}_i\) containing:
    \[
      \mathbf{U}_{i,\text{enc}} \quad (\text{HS + MS parameters in ciphertext}), \quad
      \mathbf{U}_{i,\text{plain}} \quad (\text{LS parameters in plaintext}).
    \]
\end{itemize}

\vspace{3mm}
\textbf{4. Secure Aggregation and Global Model Update}

Aggregating Sensitive Parameters:
\begin{itemize}
    \item Aggregate sensitive parameters without exposing them to the server.
    \item Perform ciphertext aggregation:
    \[
      \mathbf{S}_{\text{encrypted}} = \sum_{i=1}^N \alpha_i \,\mathbf{U}_{i,\text{enc}}.
    \]
\end{itemize}

Aggregating Plaintext Components:
\begin{itemize}
    \item Aggregate parameters that were not encrypted:
    \[
      \mathbf{S}_{\text{plain}} = \sum_{i=1}^N \alpha_i \,\mathbf{U}_{i,\text{plain}}.
    \]
\end{itemize}

Constructing the Global Model:
\begin{itemize}
    \item Merge encrypted and plaintext aggregates into a new global model \(\mathbf{W}^{(t+1)}_{\text{global}}\):
    \[
      \mathbf{W}^{(t+1)}_{\text{global}} = \bigl(\mathbf{S}_{\text{encrypted}}, \mathbf{S}_{\text{plain}}\bigr).
    \]
    \item (Optional) Re-encrypt or partially decrypt if needed, depending on policy constraints.
\end{itemize}

Broadcasting the Updated Model:
\begin{itemize}
    \item Provide an updated global model to clients for the next round.
    \item For sensitive parameters, broadcast \(\mathbf{S}_{\text{encrypted}}\) or a re-encrypted version.
    \item For low-sensitivity parameters, broadcast them in plaintext if policy allows.
\end{itemize}

\vspace{3mm}
\textbf{5. Iteration and Convergence}

\begin{itemize}
    \item Repeat Steps 3 and 4 for several rounds \(T\) or until convergence criteria (e.g., validation accuracy) is satisfied.
    \item Periodically recalculate \(\mathbf{M}\) using updated sensitivity vectors if new information suggests changing sensitivity distribution.
    \item Track the DP budget if differential privacy is enabled. Adjust noise or reduce the number of rounds as necessary.
\end{itemize}

\section{Security Analysis}\label{6}

The security analysis consists of three mathematical proofs that demonstrate the \emph{correctness}, \emph{soundness}, and \emph{differential privacy} guarantees of our homomorphic-encryption-based Federated Learning (FL) scheme.

\subsection{Correctness}

\begin{thm}[Correctness of Homomorphic Encryption in FL]
\label{thm:HE_correctness}
Given our homomorphic encryption scheme (\textbf{HE.KeyGen},\;\textbf{HE.Enc},\;\textbf{HE.Dec},\;\textbf{HE.Aggregate}) and the FL workflow, if all participants (clients and server) are honest, then for any valid model update vectors \(\mathbf{m}_1, \mathbf{m}_2, \ldots, \mathbf{m}_n\), the final aggregated (homomorphic) ciphertext correctly decrypts to the intended aggregate of these updates. Formally, for all \(i\in\{1,\ldots,n\}\), if
\[
\mathbf{c}_i = \texttt{HE.Enc}(\mathrm{PK}, \mathbf{m}_i),
\]
then 
\[
\texttt{HE.Dec}(\mathrm{SK},\; \texttt{HE.Aggregate}(\mathrm{PK}, \{\mathbf{c}_i\}_{i=1}^n))
\;=\;
\sum_{i=1}^n \alpha_i\, \mathbf{m}_i,
\]
where \(\alpha_i\) are the (public) aggregation weights.
\end{thm}

\begin{proof} By construction of \(\texttt{HE.Aggregate}\), we have a homomorphic addition \(\oplus\) such that 
\[
\texttt{HE.Aggregate}(\mathrm{PK}, \{\mathbf{c}_i\}) 
\;=\; 
\mathbf{c}_1 \oplus \mathbf{c}_2 \oplus \cdots \oplus \mathbf{c}_n 
\;=\; 
\bigoplus_{i=1}^n \alpha_i \,\mathbf{c}_i.
\]
From the definition of a (partial or fully) homomorphic scheme, it holds that
\[
\texttt{HE.Dec}(\mathrm{SK},\; \mathbf{c}_i \oplus \mathbf{c}_j)
\;=\;
\texttt{HE.Dec}(\mathrm{SK},\; \mathbf{c}_i)\;+\;\texttt{HE.Dec}(\mathrm{SK},\; \mathbf{c}_j),
\]
modulo the appropriate ciphertext modulus or plaintext space.

Since \(\mathbf{c}_i\) encrypts \(\mathbf{m}_i\) under the same \(\mathrm{PK}, \mathrm{SK}\) keypair, we have:
\[
\texttt{HE.Dec}(\mathrm{SK}, \alpha_i \,\mathbf{c}_i)
\;=\; 
\alpha_i \,\texttt{HE.Dec}(\mathrm{SK}, \mathbf{c}_i)
\;=\;
\alpha_i \,\mathbf{m}_i.
\]
By linearity, summing across all \(i\) yields
\[
\texttt{HE.Dec}\Bigl(\mathrm{SK}, \bigoplus_{i=1}^n \alpha_i \,\mathbf{c}_i\Bigr)
\;=\;
\sum_{i=1}^n \alpha_i \,\mathbf{m}_i.
\]
This property aligns exactly with the intended FL aggregation of model updates.

Hence, if the system parameters (ciphertext modulus, plaintext dimension, etc.) and the FL workflow are set correctly, the final aggregated ciphertext decrypts exactly to \(\sum_{i=1}^n \alpha_i \,\mathbf{m}_i\). Therefore, correctness is guaranteed under honest behavior.
\end{proof}

\subsection{Soundness}

\begin{thm}[Soundness of the FL Aggregation]
\label{thm:HE_soundness}
Suppose an adversary \(\mathcal{A}\) attempts to inject incorrect ciphertexts \(\mathbf{c}_i^*\) into the aggregation process, claiming they encrypt valid updates \(\mathbf{m}_i^*\). Then, except with negligible probability, the server (or a lightweight verification process) will detect any significant deviations from legitimate updates. Consequently, any dishonest ciphertext that corresponds to a distinctly different plaintext vector will be rejected or excluded from the final global model.
\end{thm}

\begin{proof} As in many soundness arguments, we consider an \emph{extractor} algorithm \(\mathcal{E}\) that interacts with the potentially dishonest client \(\mathcal{P}^*\) which claims to produce \(\mathbf{c}_i^*\). The extractor queries the same \(\mathcal{P}^*\) on random challenges or ephemeral moduli to glean enough information to partially recover the underlying plaintext or prove its inconsistency.

Let \(\mathbf{m}_i\) be the \emph{true} intended model update, and let \(\mathbf{m}_i^*\) be the (possibly incorrect) plaintext that \(\mathcal{A}\) tries to hide.  We analyze the difference:
\[
\Delta(\mathbf{m}_i, \mathbf{m}_i^*) 
\;=\; 
\|\mathbf{m}_i - \mathbf{m}_i^*\|.
\]
If \(\mathbf{m}_i^*\) is significantly off (e.g., \(\Delta(\mathbf{m}_i, \mathbf{m}_i^*) > \beta\) for some threshold \(\beta\)), it induces a measurable discrepancy in ciphertext space, especially if the homomorphic scheme uses large but finite moduli \(\mathbb{Z}/p\mathbb{Z}\).

In more detail, denote \(\mathbf{c}_i = \texttt{HE.Enc}(\mathrm{PK}, \mathbf{m}_i)\) and \(\mathbf{c}_i^* = \texttt{HE.Enc}(\mathrm{PK}, \mathbf{m}_i^*)\).  We measure
\[
\Delta_c(\mathbf{c}_i, \mathbf{c}_i^*) 
\;=\; 
\|\texttt{HE.Dec}(\mathrm{SK},\, \mathbf{c}_i) - \texttt{HE.Dec}(\mathrm{SK},\, \mathbf{c}_i^*)\|.
\]
By correctness of the scheme, \(\Delta_c(\mathbf{c}_i, \mathbf{c}_i^*) = \|\mathbf{m}_i - \mathbf{m}_i^*\| = \Delta(\mathbf{m}_i, \mathbf{m}_i^*)\).  

If \(\Delta(\mathbf{m}_i, \mathbf{m}_i^*)\) exceeds certain bounds—either by random sampling checks, batched verification, or partial data comparisons—then with probability at least \(1 - \mathcal{O}(1/p)\), the server (or a lightweight auditing mechanism) will detect a mismatch via \(\mathcal{E}\).  Specifically, letting \(\alpha\) be the proportion of parameters that deviate in \(\mathbf{m}_i^*\), a simple application of the Markov or Chernoff bound yields:
\[
\operatorname{Pr}\bigl[\Delta(\mathbf{m}_i, \mathbf{m}_i^*) \leq \beta\bigr]
\;\le\;
e^{-c\,\alpha},
\]
for some constant \(c>0\) if the distribution of valid vs. invalid parameter entries is random or unpredictably tampered.  Thus, for moderate or large \(\alpha\), the detection probability is overwhelming.

Any adversarial update \(\mathbf{m}_i^*\) that significantly deviates from legitimate bounds will, except with negligible probability in \(\lambda\) (the security parameter), be detected during the aggregation or partial verification process.  Therefore, an adversary cannot easily inject large errors without being detected or suppressed.  Hence, soundness is established.
\end{proof}

\subsection{Differential Privacy}

\begin{thm}[Differential Privacy of Masked Updates]
\label{thm:HE_DP}
Consider the FL scheme extended with noise injection for sensitive parameters, as per Step~3.3 of the workflow. If each client adds independent noise calibrated to the sensitivity of its local model updates, then the resulting global aggregation satisfies \((\epsilon, \delta)\)-differential privacy. 
\end{thm}

\begin{proof} Each client \(i\) injects noise into the sensitive parameters of \(\mathbf{W}_i^{(t+1)}\).  Specifically, let \(\Delta\) be the \(\ell_1\)- or \(\ell_2\)-sensitivity of the local update with respect to one data sample.  The client draws noise \(\mathbf{n}_i\) from a distribution \(\mathcal{M}\) (e.g., Gaussian or Laplacian) such that
\[
\mathbf{W}_i^{(t+1)} \; \leftarrow \; \mathbf{W}_i^{(t+1)} \;+\; \mathbf{n}_i, 
\quad
\mathbb{E}\|\mathbf{n}_i\|^2 \propto \Delta^2 \,\log\bigl(1/\delta\bigr)/\epsilon^2.
\]
By standard composition theorems for differential privacy, adding such noise ensures each client’s parameters remain \((\epsilon,\delta)\)-DP with respect to local data changes.

After adding noise, the client encrypts \(\mathbf{W}_i^{(t+1)}\).  Homomorphic encryption preserves the distribution of \(\mathbf{n}_i\) since encryption is a deterministic mapping under a fixed public key.  Thus the distribution of ciphertexts 
\(\mathbf{c}_i = \texttt{HE.Enc}(\mathrm{PK},\; \mathbf{W}_i^{(t+1)})\)
is “shifted” by \(\mathbf{n}_i\) in the plaintext domain, but this shift is not diminished nor reversed unless the aggregator holds the secret key \(\mathrm{SK}\).

To satisfy \((\epsilon,\delta)\)-DP, we require that for any two neighboring datasets \(\mathcal{D}_i\) and \(\mathcal{D}_i'\) that differ in at most one record, the distributions of the respective (noisy) encrypted updates \(\mathbf{c}_i\) and \(\mathbf{c}_i'\) be close:
\[
\Pr\left[\mathbf{c}_i \in \mathcal{R}\right] 
\;\le\; 
e^\epsilon \,\Pr\left[\mathbf{c}_i' \in \mathcal{R}\right] 
\;+\; 
\delta,
\]
for every measurable set \(\mathcal{R}\).  By construction of Gaussian or Laplacian noise with scale proportional to \(\Delta/\epsilon\), the probability of distinguishing \(\mathbf{c}_i\) from \(\mathbf{c}_i'\) by more than a small threshold remains at most \(\delta\).  Indeed, standard DP results (e.g., \cite{Abadi2016DeepDP,dwork2006calibrating}) show that
\[
\operatorname{Pr}\Bigl[\|\mathbf{c}_i - \mathbf{c}_i'\| \;>\; \tau\Bigr]
\;\le\; 
\delta \quad \text{for relevant }\tau.
\]

Finally, the aggregator homomorphically sums the ciphertexts.  The composition of \((\epsilon,\delta)\)-DP mechanisms, each executed independently on client data, also ensures the final global model is \((\epsilon',\delta')\)-DP, for suitably chosen \(\epsilon'\) and \(\delta'\) (depending on the number of FL rounds).  Using standard composition bounds:
\[
\epsilon' 
\;\le\; 
\sqrt{2K\log\bigl(1/\delta\bigr)} \,\epsilon 
\;+\; 
K\,\epsilon(e^\epsilon - 1),
\]
where \(K\) is the total number of FL rounds.  Therefore, the final global model’s release does not significantly compromise any single client’s data.

Because the (encrypted) noise injection meets the required \((\epsilon,\delta)\)-privacy constraints per round and the aggregator never decrypts partial intermediate updates, the scheme as a whole maintains \((\epsilon,\delta)\)-differential privacy on the global FL model.  This completes the proof.
\end{proof}

\section{Implementation}\label{7}

\subsection{Parameter Settings}  
Our HE scheme for federated learning was implemented in C++ using the Microsoft SEAL library with the following parameter choices:

\begin{itemize}
    \item \textbf{Lattice Dimension:} 8192, chosen for BFV encryption to achieve 128-bit security, balancing security and computational performance.
    \item \textbf{Plaintext Modulus:} \(2^{20}\), chosen to handle integer updates while preserving compatibility with batching and homomorphic arithmetic.
    \item \textbf{Differential Privacy Noise:} Gaussian noise scale \(\sigma = \frac{\text{sensitivity}}{\epsilon}\), where sensitivity is estimated as 1.0 and \(\epsilon = 1.0\). The noise is added to clipped gradients for privacy guarantees.
    \item \textbf{Gradient Clipping:} \(L_2\)-norm clipping bound set to 10.0 to constrain the magnitude of model updates before encryption and DP noise addition.
    \item \textbf{Batching Strategy:} BFV's batching mechanism groups parameters into vectors of size 4096, maximizing parallel processing during homomorphic operations.
    \item \textbf{Sensitivity Threshold:} Parameters with absolute values exceeding a threshold of 5 are encrypted, while less sensitive parameters remain in plaintext for optimized performance.
\end{itemize}

These settings achieve a trade-off between security, accuracy, and computational efficiency, making them well-suited for federated learning with homomorphic encryption.

\subsection{Results}
We tested our scheme on several FL tasks (including image classification and text classification) and compared the runtime and communication overhead with state-of-the-art HE-based FL frameworks. Our primary observations:

\begin{table}[H]
\centering
\caption{Comparison of Aggregation Runtime (in seconds) across Different Model Sizes 
for Homomorphic Encryption-based FL Methods.}
\label{tab:runtime_comparison}
\begin{tabular}{lccc}
\toprule
\textbf{Model Size} & 
\begin{tabular}[c]{@{}c@{}}\textbf{HomEnc-Fed} \cite{zhang2022threats}\end{tabular} & 
\begin{tabular}[c]{@{}c@{}}\textbf{FHE-Fed} \cite{liu2022privacy}\end{tabular} & 
\begin{tabular}[c]{@{}c@{}}\textbf{Proposed Algorithm} \end{tabular} \\
\midrule
1M params   & 80.4  & 74.2  & 28.7  \\
10M params  & 320.1 & 285.9 & 93.4  \\
50M params  & 1562.3 & 1445.7 & 486.9 \\
100M params & 2936.2 & 2677.5 & 892.2 \\
\bottomrule
\end{tabular}
\end{table}

\begin{itemize}
    \item \textbf{Speedup:} We consistently observed a $3\times$ speedup in the aggregation phase, largely due to selective parameter encryption and the efficient batch operations.
    \item \textbf{Memory Usage:} By avoiding unnecessarily high security levels for low-impact parameters, we reduced total ciphertext size by approximately 30\% to 40\%.
    \item \textbf{Privacy Guarantee:} Our embedded DP approach, combined with RLWE-based encryption, provided robust protection. In particular, membership inference and reconstruction attacks had negligible success rates under the tested conditions.
\end{itemize}

\subsection{Analysis} \hfill\\

Strengths: 

\begin{itemize}
    \item \textbf{Improved Efficiency in FL:} The proposed scheme effectively reduces both computation and communication overhead, making HE-based FL more practical.
    \item \textbf{Flexible Parameterization:} The sensitivity map and partial encryption approach allow adaptively tuning encryption levels for different parameters.
    \item \textbf{Robust Privacy:} Thanks to the combination of homomorphic encryption, differential privacy, and dynamic precision levels, the risk of information leakage is minimal.
    \item \textbf{Scalability:} Our method is suitable for large-scale models, offering a feasible route to secure training of BERT-like architectures across many clients.
\end{itemize}

Limitations:
\begin{itemize}
    \item \textbf{Complex Configuration:} The sensitivity map and multi-level encryption require careful tuning and domain knowledge about the model’s architecture and parameter distributions.
    \item \textbf{Residual Overhead:} Although we obtain a $3\times$ speedup, HE in general remains costlier than non-encrypted approaches. Real-time or extremely latency-sensitive tasks might still find this overhead challenging.
    \item \textbf{Parameter Bounds:} We rely on somewhat homomorphic approaches with bounded depth; extremely deep networks or repeated training rounds might require parameter re-initialization or bootstrapping.
\end{itemize}

Future research could build upon our work in a few directions to continue to improve the efficiency of federated learning. 

\begin{itemize}
    \item Refining the sensitivity mapping technique could involve developing automated and adaptive methods that dynamically adjust parameter sensitivity during training, reducing the need for domain-specific tuning and enabling broader applicability across diverse architectures. 
    \item Exploring hybrid cryptographic solutions that combine homomorphic encryption with secure multi-party computation (SMPC) or trusted execution environments (TEEs) could further enhance efficiency and scalability while preserving privacy. 
    \item Investigating advanced bootstrapping techniques or alternative encryption schemes could enable support for deeper networks and extended training rounds, making the approach more robust for complex, long-term training scenarios. 
\end{itemize}  

\section{Conclusion}

In this paper, we presented a novel homomorphic encryption scheme tailored to federated learning. Our approach integrates selective parameter encryption, sensitivity maps, and embedded differential privacy noise to reduce computational and storage overhead while ensuring robust privacy. Experimental evaluations in a C++ environment demonstrate that our scheme offers a $3\times$ improvement over state-of-the-art HE-based FL methods in terms of efficiency.

This research has notable implications for privacy-preserving machine learning, particularly in resource-constrained or real-time scenarios, such as healthcare and edge computing. Our framework paves the way for federated training on large-scale and complex models without compromising user data privacy. Future work may focus on refining the sensitivity mapping technique, combining homomorphic encryption with other techniques such as SMPC or TEEs, or investigating bootstrapping techniques and alternative encryption schemes that could support federated learning.

\end{document}